\documentclass[10pt]{article}

\usepackage{amssymb,amsmath, amsthm, setspace, color, graphics}
\usepackage{dsfont}              
\usepackage[dvips]{graphicx}
\usepackage[english]{babel}
\usepackage{float}
\usepackage{caption}
\usepackage{subcaption}
\usepackage[authoryear]{natbib}
\usepackage[pdfborder={0 0 0},final=true,colorlinks=true,linkcolor=magenta,citecolor=blue]{hyperref}
\numberwithin{equation}{section}
\newcommand{\tab}{\hspace*{2em}}              
\theoremstyle{plain}
\newtheorem{thm}{Theorem}[section]
\newtheorem{cor}[thm]{Corollary}
\newtheorem{lem}[thm]{Lemma}
\newtheorem{prop}[thm]{Proposition}
\theoremstyle{plain}
\newtheorem{assumption}[thm]{Assumption}      
\theoremstyle{definition}

\theoremstyle{definition}

\theoremstyle{remark}
\newtheorem{rem}{Remark}[section]
\begin{document}
\singlespace

\title{\vspace{-6ex} Active Portfolio Management, Positive Jensen-Jarrow Alpha, and Zero Sets of CAPM}
\author{ Godfrey Cadogan \thanks{ Corresponding address: Information Technology in Finance, Institute for Innovation and Technology Management, Ted Rogers School of Management, Ryerson University, 575 Bay, Toronto, ON M5G 2C5; e-mail: \textcolor[rgb]{0.00,0.00,1.00}{\href{mailto:godfrey.cadogan@ryerson.ca}{godfrey.cadogan@ryerson.ca}}.  This paper benefited from conversations with John A. Cole. Research support from the Institute for Innovation and Technology Management is gratefully acknowledged. All errors which may remain are my own.} \\Working Paper\\ Comments welcome}
\renewcommand\thefootnote{\fnsymbol{footnote}}
\date{\today\vspace{-4ex}}
\maketitle
\renewcommand\thefootnote{\arabic{footnote}}
\begin{abstract}
We present conditions under which positive alpha exists in the realm of active portfolio management-- in contrast to the controversial result in \cite[pg.~20]{Jarrow2010} which implicates delegated portfolio management by surmising that positive alphas are illusionary.  Specifically, we show that the critical assumption used in \cite[pg.~20]{Jarrow2010}, to derive the illusionary alpha result, is based on a zero set for CAPM with Lebesgue measure zero. So conclusions based on that assumption may well have probability measure zero of occurrence as well. Technically, the existence of [Tanaka] local time on a zero set for CAPM implies existence of positive alphas. In fact, we show that positive alpha exists under the same scenarios of ``perpetual event swap" and ``market systemic event" \cite{Jarrow2010} used to formulate the illusionary positive alpha result. First, we prove that as long as asset price volatility is greater than zero, systemic events like market crash will occur in finite time almost surely. Thus creating an opportunity to hedge against that event. Second, we find that Jarrow's ``false positive alpha" variable constitutes  portfolio manager reward for trading strategy. For instance, we show that positive alpha exists if portfolio managers develop hedging strategies based on either (1) an exotic [barrier] option on the underlying asset--with barrier hitting time motivated by the ``market systemic" event, or (2) a \emph{swaption} strategy for the implied interest rate risk inherent in Jarrow's triumvirate of riskless rate of return, factor sensitivity exposure, and constant risk premium for a perpetual event swap.
\\
\\
\emph{Keywords:} empirical alpha process, portfolio hedging strategies,  active portfolio management, market systemic risk, swaption, delegated portfolio management, local time of alpha
\\
\\
\emph{JEL Classification Codes:} C02, G12, G13\\
\emph{Mathematics Subject Classification 2010:} 91G10, 91G70
\end{abstract}
\singlespace
\renewcommand\thefootnote{\fnsymbol{footnote}}
\renewcommand\thefootnote{\arabic{footnote}}
\tableofcontents
\listoffigures
\newpage
\section{Introduction}\label{sec:Intro}
Portfolio performance evaluation is one of the most actively researched areas in financial economics. Typically, such evaluation is based on the statistical properties of the intercept term--alpha--in a linear asset pricing model that subsumes the capital asset pricing model (CAPM)\footnote{See \cite{Sharpe1964}.}. There, assuming zero transaction costs, alpha is construed as portfolio returns net of risk adjusted returns measured relative to a standard index like the S\&P 500 or other mutually agreed benchmark\footnote{See \cite[pg.~88]{GrinoldKahn2000}.}. Positive alpha implies that a portfolio manager's investment decisions increased the returns on a portfolio while the converse is true for negative alpha.  However, empirical tests of portfolio performance evaluation tend to produce mixed results\footnote{See e.g., \cite[pp.~2-3]{FersonLin2010} for a recent survey of the literature.}. The goal of this paper is to provide an explanation for the existence of positive alpha by and through a continuous trading strategy representation theorem introduced in \cite{Cadogan2011b}. Contrasted against a controversial result in \cite{Jarrow2010} who surmised that positive alphas are illusionary based on a continuous asset pricing model constructed in \cite{JarrowProtter2010}.

\tab To be sure, the genesis of the positive alpha controversy stems from seminal papers by \cite{Jensen1967} and \cite{TreynorMazuy1966} who proposed performance evaluation methods motivated by the single factor CAPM. There, Jensen argued that a statistically significant positive intercept (alpha) in a single factor CAPM supports superior portfolio manager performance. However, he was unable to find statistically significant evidence of positive alpha in monthly returns for his sample of mutual funds.  Treynor and Mazuy augmented the single factor CAPM with a quadratic market premium term to identify convexity in portfolio returns as a measure of superior performance due to market timing. Those methods were subsequently extended in influential papers by \cite{Merton1981} and \cite{HenrikMerton1981}, and tested empirically in \cite{Henricksson1984} and \cite{GlostenJagannathan1994}. Those authors used isomorphism between the returns on an option on the underlying asset, in lieu of the Treynor-Mazuy quadratic term, to identify a portfolio manager's market timing skills. It should be noted in passing that market timing or stock picking involves trade strategy.

\tab The advent of \cite{Ross1976} arbitrage pricing theory, \cite{Merton1973} intertemporal CAPM (``ICAPM") and its consumption based (CCAPM) extension by \cite{Breeden1979} provide theoretical foundations for extending portfolio performance evaluation to multifactor asset pricing models. Arguably, the most prominent empirical multifactor linear asset pricing model referenced in the literature is \cite{FamaFrench1993} 3-factor hedge model\footnote{Other popular hedge factor models include \cite{EltonGruberBlake1996} (4-factor model includes benchmark plus SMB, growth minus value stock returns (GMV) and bond index return (BI);
\cite{Cahart1997} ( 4-factor model includes benchmark plus SMB, HML, momentum factor (MOM)); \cite{AgrawalNaik2004} (6-factor model includes benchmark plus SMB, HML, MOM, out-of-the-money put options (OTM Put), at-the-money-put options(ATM Put)). See \cite{Nohel2010} and \cite{Goyal2012} for surveys of literature.}. Statistically significant hedge factors suggests that they are proxies for alpha relative to a given benchmark, and they portend trade strategy for alpha. Moreover, empirical analysis of high frequency data (i.e., daily and intraday data) typically generated by active portfolio management, has been more encouraging about portfolio manager market timing ability\footnote{See \cite{ChanceHemler2001, BollenBusse2001}}, and suggests that sample frequency may be the cause of earlier studies failure to find such relationships.

\tab Recent contributions by \cite{Jarrow2010}, \cite{JarrowProtter2010}, and \cite{Cadogan2011b} introduced empirical alpha processes\footnote{Technically, the Jarrow-Protter alpha is not an empirical process since it was an ad hoc introduction into their model. By contrast, the empirical alpha processes in \cite[Theorem~4.6]{Cadogan2011b} is derived from an asymptotic theory of canonical multifactor linear asset pricing models. See e.g., \cite[pg.~1]{ShorackWellner1986} for definition of empirical process.}, for a multifactor linear asset pricing model, as a means of evaluating alpha in active portfolio management in continuous time. However, there is some controversy involving the empirical processes of those authors because they appear to disagree on whether positive alphas can be identified.

\tab In particular, \cite{Jarrow2010, JarrowProtter2010} introduced a heuristic alpha pricing model, which they coigned a ``K-factor return model for an asset's evolution". To support their argument that positive alphas are illusionary, \cite[pg.~20]{Jarrow2010} introduced an \emph{event swap} to identify the impact of a so called ''phantom factor" on portfolio alpha. Additionally, he added an extraneous alpha variable to the model and then imposed identifying restrictions to ascertain its efficacy. Specifically, he assumed the possibility of a ``market-systemic" event like a stock market crash, and a financial contract with notional value \$ 1 for which the buyer pays a fixed spread $c$ above the risk free rate $r_t$ for protection $I$ from the seller in the event the ''market-systemic" event occurs. In other words, $c$ is a risk premium. This event is incorporated in the model under the notion of a long position in the ``K-th factor" $X_K(t)$ which is functionally equivalent to a \emph{perpetual event sawp} on the market-systemic event. Jarrow's  result implicates delegated portfolio management. It suggests that portfolio managers (and hence investors) are not compensated for bearing the risk of unfavorable shifts in the investment opportunity set \`{a} la \cite[pg.~882,~eq.~(34)]{Merton1973}. It also implicates \cite[pg.~883]{Merton1973} heuristics about the existence of synthetic bonds in market disequilibrium or existence of market neutral portfolios. On a more technical level, it implicates the class of controller-stopper problems introduced by \cite{KaratzasSudderth1999} who provide regularity conditions for optimal control of diffusions to a goal--such as those in contention here.

\tab By contrast, \cite{Cadogan2011b} introduced an adaptive behavioral model of portfolio manager market timing, in the context of high frequency trading, by applying asymptotic theory to a canonical multifactor linear asset pricing model, to derive a continuous time trade strategy representation theorem which generates portfolio alphas\footnote{Arguably, trade strategy representation can be included in incentive contracts. See e.g., \cite{DybvigFarnsworthCarpenter2010}.}. In particular, he used regularity conditions of linear asset pricing models to show that, for single factor models like the CAPM, so called Jensen's alpha behaves like a Brownian bridge. A result consistent with one of the empirical regularities of high frequency trading where portfolio managers make extensive use of futures markets, and start the trading day with a zero inventory position, and close out inventories at the end of the trading period\footnote{See \cite[pg.~120]{EasleyPradoOHara2011}; \cite[pg.~2]{Brogaard2010}(Using a net inventory of zero at the end of the trading day ''I estimate HFTs generate around \$2.8 billion in gross annual trading proits and on a per \$100 traded earn three-fourths of a penny.  The per-dollar traded pro?t is about one-seventh that of traditional market makers''). Compare \cite[pp.~27-28]{Hull2006}(daily settlement at end of trading day in futures markets). }.

\tab In summary, this paper compares and contrasts \cite{Jarrow2010} and \cite{JarrowProtter2010} with \cite{Cadogan2011b} in the context of the \emph{event swap} and market systemic event environment posited in \cite{Jarrow2010}. And it provides explanations for Jarrow-Protter seeming misperception(s) about the existence of positive alpha, and proceeds as follows. In \autoref{sec:Model} we introduce the empirical alpha processes in contention, and establish conditions for stochastic equivalence between them in Proposition \autoref{prop:StochasticEquivalenceModels}. In \autoref{sec:MarketSystemicAlpha} we provide analytics for positive portfolio alpha including but not limited to its local time, and the existence of nonzero Jarrow alpha due to financial engineering or otherwise. Te main result there is Proposition \autoref{prop:NonZeroJarrowAlpha}. In \autoref{sec:Conclusion} we conclude the paper.

\section{The Models}\label{sec:Model}

We begin by stating the alpha representation theorems of \cite{Cadogan2011b,JarrowProtter2010} [without proofs] in seriatim on the basis of the assumptions in \cite{Jarrow2010}. Then we show that the two theorems are stochastically equivalent--at least for a single factor--and provide some analytics.
\begin{assumption}
    Asset markets are competitive and frictionless with continuous trading of a finite number of assets.
\end{assumption}
\begin{assumption}
   Asset prices are adapted to a filtration of background driving Brownian motion.
\end{assumption}
\begin{assumption}
   Prices are ex-dividend.
\end{assumption}
\begin{thm}[Trading strategy representation. \cite{Cadogan2011b}]\label{thm:GammaRepresentationTheorem}~\\
   Let $(\Omega,\mathcal{F}_t,\mathbb{F},P)$ be a filtered probability space, and $Z=\{Z_s,\mathcal{F}_s;\; 0\leq s < \infty\}$ be a hedge factor matrix process on the augmented filtration $\mathbb{F}$. Furthermore, let $a^{(i,k)}(Z_s)$ be the $(i,k)$-th element in the expansion of the transformation matrix $(Z_s^TZ_s)^{-1}Z^T_s$, and $B=\{B(s),\mathcal{F}_s;\;s\geq 0\}$ be Brownian motion adapted to $\mathbb{F}$ such that $B(0)=x$. Assuming that $B$ is the background driving Brownian motion for high frequency trading, the hedge factor sensitivity process, i.e. trading strategy, $\gamma = \{\gamma_s,\mathcal{F}_s;0\leq s < \infty\}$ generated by portfolio manager market timing for Brownian motion starting at the point $x\geq 0$ has representation
   \begin{equation*}
      d\gamma^{(i)}(t,\omega) = \sum^j_{k=1}a^{(i,k)}(Z_t)\biggl[\frac{x}{1-t}\biggr]dt  - \sum^j_{k=1}a^{(i,k)}(Z_t)dB(t,\omega),\;x\geq 0
   \end{equation*}
   for the $i$-th hedge factor $i=1,\dotsc,p$, and $0\leq t\leq 1$. \hfill $\Box$
\end{thm}
\begin{proof}
   See \cite[Thm.~4.6]{Cadogan2011b}.
\end{proof}
\cite{Cadogan2011b} computes an alpha vector $\pmb{\alpha}(t,\omega)=Z\pmb{\gamma}(t,\omega)$ where $Z$ is a matrix of hedge factors and $\pmb{\gamma}(t,\omega)$ is a vector of hedge factor sensitivity that constitutes the active portfolio manager's trading strategy. This parametrization is consistent with hedge factor models\footnote{See \cite{Nohel2010} for a review.} such as \cite{FamaFrench1993}.
\begin{thm}[(\cite{JarrowProtter2010})]\label{thm:JarrowProtterAlpha}~\\
    Given no arbitrage, there exist $K$ portfolio price processes $\{X_1(t),\dotsc,X_K(t)\}$ such that an arbitrary security's excess return with respect to the default free spot interest rate, i.e. risk free rate, $r_t$ can be written
    \begin{equation}
       \frac{dS_i(t)}{S_i(t)} - r_tdt = \sum^K_{k=1}\beta_{ik}(t)\biggl(\frac{dX_i(t)}{X_i(t)}-r_tdt \biggr) +\delta_i(t)d\epsilon_i(t)
    \end{equation}
    where $\epsilon_i(t)$ is a Brownian motion under $(P,\mathcal{F}_t)$ independent of $X_k(t)$ for $k=1,\dotsc,K$ and $\delta_i(t)=\sum^d_{k=K+1}\sigma^2_{ik}(t)$ \hfill $\Box$
\end{thm}
\begin{proof}
   See \cite[Thm.~4]{JarrowProtter2010}.
\end{proof}
In contrast to \cite{Cadogan2011b} endogenous alpha, \cite[pg.~18]{Jarrow2010}  surmised that ``[i]n active portfolio management, an econometrician would add a constant $\alpha_i(t)dt$ to [the right hand side of] this model". So \cite{Cadogan2011b} alpha is adaptive, whereas \cite{Jarrow2010} alpha is exogenous. Without loss of generality we assume that the filtration in each model is with respect to background driving Brownian motion.
\subsection{Trade strategy representation of alpha in single factor CAPM}\label{subsec:TradestrategyAlpha}
If our representation theory is well defined, then it should shed light on the behavior of \emph{alpha} in a single factor model like CAPM where there is no hedge factor. In particular, for \emph{cumulative alpha} $A(t)$ let\footnote{See \cite[pp.~121-122]{ChristophFersonGlassman1998} for representation of alpha conditioned on a $Z$-matrix of economic information.}
\begin{align}
   A(t) &= Z\pmb{\gamma}(t)\\
   \intertext{where $Z$ is a hedge factor matrix, and $\pmb{\gamma}(t)$ is a trade strategy vector. Thus, for some vector $\pmb{\alpha}(t)$ we have}
   d\{A(t)|Z\} &= \pmb{\alpha}(t)dt = Zd\pmb{\gamma}(t).\\
   \text{Let}\;Z &=\mathds{1}_{\{n\}}\\
   \intertext{So that we have}
   (Z^TZ)^{-1}Z^T &= n^{-1}\mathds{1}_{\{n\}}^T\;\;\text{and}\;\;a^{(1,k)}(Z_s) = n^{-1},\;k=1,\dotsc,n\\
   \intertext{Substitution of these values in Theorem \autoref{thm:GammaRepresentationTheorem} gives us, by abuse of notation, the scalar equation for \emph{trade strategy alpha}}
   -\alpha(t)dt &=-d\gamma^{(1)}(t) = -\frac{x}{1-t}dt + dB(t)\label{eq:CadoganSingleFactorAlpha}
\end{align}
That is the equation of a Brownian bridge starting at $B(0)=x$ on the interval $[0,1]$. See \cite[pg.~268]{KarlinTaylor1981}. \autoref{fig:EmpiricalAlphaFig1} depicts 5-sample paths for a Brownian bridge over the interval $[0,1]$ based on the Matlab code in \autoref{apx:MatlabodeForBB}.
\begin{figure}
   \centering
      \begin{minipage}[h]{0.4\linewidth}
         \captionof{figure}{Alpha Representation}
         \label{fig:EmpiricalAlphaFig1}
         \centerline{\includegraphics[scale=.5]{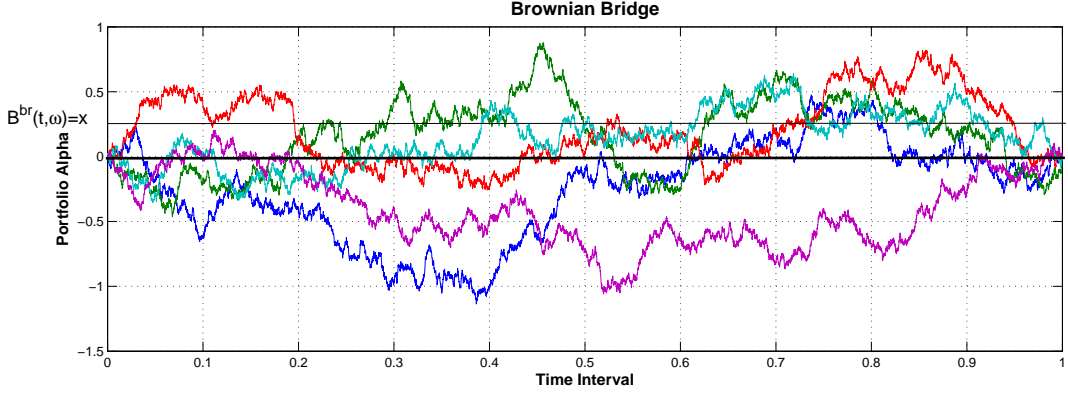}}
      \end{minipage}
   \hspace{1cm}
\end{figure}
So that trade strategy alpha in \eqref{eq:CadoganSingleFactorAlpha} is given by
\begin{align}
   d\gamma^{(1)}(t) &= -dB^{br}(t)\\
   \gamma^{(1)}(t) &= -B^{br}(0)+B^{br}(t)\label{eq:GammaEstimate}\\
   \intertext{However, there is more. According to Girsanov's formula in \cite[pg.~162]{Oksendal2003}, we have an equivalent probability measure $Q$ based on the martingale transform}
   M(t,\omega) &= \exp\biggl(\int^t_0\frac{x}{1-s}dB(s,\omega)-\int^t_0\biggl(\frac{x}{1-s}\biggr)^2 ds\biggr)\\
   dQ(\omega) &= M(T,\omega)dP(\omega),\quad 0\leq t\leq T\leq 1\\
   \intertext{Furthermore, we have the $Q$-Brownian motion}
   \hat{B}(t) &= -\int^t_0 \frac{x}{1-s}ds + B(t),\quad \text{and}\\
   d\gamma^{(1)}(t) &= -d\hat{B}(t)=-dB^{br}(t)\\
   \intertext{In other words, $\gamma^{(1)}$ is a $Q$-Brownian motion--in this case a Brownian bridge--that reverts to the origin starting at $x$. We note that}
   \frac{d\gamma^{(1)}(t)}{dt} &= -\frac{dB^{br}(t)}{dt} = \tilde{\epsilon}_t\\
   \intertext{Hence the ''residual(s)" $\tilde{\epsilon}_t$, associated with rate of change of Jensen's alpha, have an approximately skewed U-shape pattern for a Brownian bridge over $[0,1]$. Additionallly, \cite[pg.~270]{KarlinTaylor1981} show that the Brownian bridge can be represented by the Gaussian process}
   G(t) &= (1-t)B\biggl(\frac{t}{1-t}\biggr)\\
   E^Q[G(s_1)G(s_2)] &=
   \begin{cases}
      s_1(1-s_2)  & s_1 < s_2 \\
      s_2(1-s_1)  & s_2 < s_1
   \end{cases}\\
   \intertext{This is the so called \emph{Doob transformation}, see \cite[pp.~394-395]{Doob1949}, and}
   B(t) &= (1+t)B\biggl(\frac{t}{1+t}\biggr)\\
   \intertext{\cite[Thm. ~3.4.6,~pg.~174,~and~pg.~359]{KaratzasShreve1991} also provide further analytics which show that we can write the trade strategy alpha as}
   d\gamma^{(1)}(t) &= \frac{1-\gamma^{(1)}}{1-t}dt+dB(t);\quad 0\leq t\leq 1,\;\;\gamma^{(1)}=0\label{eq:GammaRepresentationAlternative}\\
   M(t) &= \int^t_0\frac{dB(s)}{1-s}\\
   T(s) &= \inf\{t| <M>_t > s\}\\
   M(t) &= B_{<M>_{T(t)}},\quad B_s = M_{T(s)}\label{eq:DambisDubinsSchwarz}
\end{align}
where $<M>_t$ is the quadratic variation\footnote{See \cite[pg.~31]{KaratzasShreve1991}.} of the local martingale $M$. The time changed process above is consistent with the clock time and trading or business time dichotomy reported in \cite[pg.~138]{Clark1973} and \cite[pp.~117-118]{CarrWu2004}.
\subsection{Local time of alpha for single factor CAPM}\label{subsec:CAPM_LocalTime}
Under the Dambis-Dubins-Schwarz criteria in \eqref{eq:DambisDubinsSchwarz}, alpha is a time changed martingale--in this case Brownian motion. In the absence of a hedge factor, the single factor or benchmark, is perfectly tracked if
\begin{align}
   \gamma^{(1)}(t) &= 0\\
   \intertext{Whereupon in \ref{eq:GammaEstimate} we must have }
   \hat{B}(t,\omega) &= x
\end{align}
This reduces the problem to one of local time [of a Brownian bridge] at $x$ depicted in \autoref{fig:EmpiricalAlphaFig1} on page \pageref{fig:EmpiricalAlphaFig1}. We can think of $x$ as a hurdle rate such as transaction costs that the manager must attain to break even. The Lebesgue measure associated with the level set
\begin{align}
   \mathcal{Z} &= \{\omega|\;\hat{B}(t,\omega)=x\}\label{eq:AlphaLevelSet}
   \intertext{is such that}
   Q\{\mathcal{Z}\} &= 0
\end{align}
So that a perfectly tracked benchmark--one with zero alpha--has Lebesgue measure zero. Whereupon single factor equilibrium asset pricing models like the CAPM exists on a set of measure zero. In other words, there is market disequilibrium, and hence arbitrage opportunities, almost surely\footnote{See e.g. \cite{ShleiferVishny1997} (arbitrage cannot be priced away for sufficiently large mispricing away from fundamentals)}. Thus, there is a non-zero probability to build a so called Dutch book against a tracked benchmark portfolio\footnote{\cite[pg.~19]{Jarrow2010} critiqued the existence of positive alphas on the basis of a formulation diferent from ours}. In particular, \emph{Tanaka`s} formula\footnote{See \cite[pg.~205]{KaratzasShreve1991}} tells us that the local time $L(t,x,\omega)$ of \emph{alpha}\footnote{Recall that $\gamma^{(1)}=\hat{B}(0)-\hat{B}(t)$ in \ref{eq:GammaEstimate}.} on $\mathcal{Z}$ is given by
\begin{align}
   2L(t,x,\omega) &= |\hat{B}(t,\omega)-x|-|\mathfrak{z}-x|-\int^t_0\text{sgn}(\hat{B}(s,\omega)-x)d\hat{B}(s,\omega)
\end{align}
for some $\mathfrak{z}\in \mathbb{R}$\footnote{\cite[pp.~337-338]{DellacherieMeyer1982} provides details on Tanaka formula.}. In other words, even though the CAPM has measure zero probability, its local time exists. Perhaps more important, the perfectly hedged portfolio problem, i.e. the CAPM problem, reduces to one of stochastic optimal control--guiding $\gamma^{(1)}$ to a goal of $0$ by keeping it as close to $0$ as possible. This problem, and related ones, were solved by \cite{BenesSheppWitsen1980} and in \cite[Chapter~6.2]{KaratzasShreve1991}. Thus, our \emph{alpha representation} vector in Theorem \ref{thm:GammaRepresentationTheorem} subsumes \ref{eq:GammaRepresentationAlternative}, and shed new light on equilibrium asset pricing models like the CAPM, and explain why Jensen's alpha is oftentimes negative in the empirical literature. To wit, negative [positive] Jensen's alpha is an artifact of its Brownian bridge feature in \ref{eq:GammaEstimate} according as $B^{br}(0)$ is less [greater] than $B^{br}(t)$. It is always reverting back to an equilibrium zero that has probability zero of absorbtion. We close this subsection with the following.
\begin{prop}[Zero sets and local time of single factor CAPM]\label{prop:CAPM_ZeroSetLocalTime}~\\
   The single factor CAPM exists on a set $\mathcal{Z}=\{\omega|\;\alpha(t,\omega)=0\}$ with Q-measure zero, and $x=0$ in Tanaka's formula, where $\alpha(t,\omega)=\hat{B}(t,\omega)$ with local time $$2L(t,0,\omega) = |\hat{B}(t,\omega)|-|\mathfrak{z}|-\int^t_0\text{sgn}(\hat{B}(s,\omega))d\hat{B}(s,\omega)$$ where $\mathfrak{z}\in \mathbb{R}$ and $\hat{B}(s,\omega)$ is a Brownian bridge. \hfill $\Box$
\end{prop}
\begin{cor}[Tracking error problem for augmented CAPM]\label{cor:CAPM_TrackingError}~\\
    The tracking error problem for a perfectly hedged augmented CAPM is functionally equivalent to guiding alpha to a goal of zero. \hfill $\Box$
\end{cor}

\subsection{Stochastic equivalence between Jensen-Jarrow alpha and\\ trade strategy alpha}\label{subsec:StochasticEquivAlpha}

We use a single factor representation of Theorem \autoref{thm:GammaRepresentationTheorem} and Theorem \autoref{thm:JarrowProtterAlpha} with $K=1$ for comparison as follows\footnote{\cite[pg.~12,~eq.~(12)]{JarrowProtter2010} refer to the ensuing as a 'regression equation" for which an econometrican tests the null hypothesis $H_0: \alpha(t)dt=0$. However, \cite[eq.~2.1]{Cadogan2011b} applied asymptotic theory to econometric specification of a canonical multifactor linear asset pricing model augmented with portfolio manager trading strategy to identify portfolio alpha.}. Using \cite[eq(5),~pg.~18]{Jarrow2010} formulation, let
\begin{align}
    \alpha(t)dt &= \frac{dS(t)}{S(t)} - r_tdt -\beta_1(t)\biggl(\frac{dX_1(t)}{X_1(t)} - r_tdt\biggr)-\sigma dB(t)\\
    &= \frac{dS(t)}{S(t)}-\beta_1(t)\frac{dX_1(t)}{X_1(t)}-r_t(1-\beta_1(t))dt - \sigma dB(t)\label{eq:JarrowProtterSingleFactorAlpha}
\end{align}
Comparison with \ref{eq:CadoganSingleFactorAlpha} suggests that the two alphas are equivalent if the following restrictions are imposed
\begin{align}
    &d\gamma^{(1)}(t)\equiv\alpha(t)dt\\
    &\Rightarrow\frac{dS(t)}{S(t)}-\beta_1(t)\frac{dX_1(t)}{X_1(t)} = 0\label{eq:Self_Fin_Portfolio}\\
    &\sigma = 1\\
    &\frac{x}{1-t} = (\beta_1(t)-1)r_t
\end{align}
So that
\begin{align}
    -d\gamma^{(1)}&=-(\beta_1(t)-1)r_t dt+dB(t)\label{eq:JarrowCadoganEquivalence}\\
    \intertext{In fact, if for some constant drift $\mu_X$, volatility $\sigma_X$, and $P$-Brownian motion $B_X$ we specify the ''hedge factor" dynamics}
    \frac{dX_1(t)}{X_1(t)} &= \mu_Xdt + \sigma_X dB_X\label{eq:Factor1Dynamics}\\
    \intertext{then after applying Girsanov's change of measure to \ref{eq:Factor1Dynamics} we can rewrite \ref{eq:Self_Fin_Portfolio} as}
    \frac{dS(t)}{S(t)} &= \beta_1(t)\sigma_X d\tilde{B}_X(t)\label{eq:FuncEquivAssetPriceDynamics}
\end{align}
for some $Q$-Brownian motion $\tilde{B}_X$. These restrictions, required for functional equivalence between the two models, are admissible and fairly mild. We summarize the forgoing in the following
\begin{prop}[Stochastic equivalence of alpha in single factor models.]\label{prop:StochasticEquivalenceModels}
   Assume that asset prices are determined by a single factor linear asst pricing model. Then the \cite{JarrowProtter2010} return model is stochastically equivalent to \cite{Cadogan2011b} trading strategy representation model. \hfill $\Box$
\end{prop}
\begin{cor}[\cite{Jarrow2010} trade strategy drift factor]\label{cor:JarrowTradeStrategy}~\\
   $\{\beta_1(t)\}_{t\in[0,T]}$ in \cite{Jarrow2010} is a trade strategy drift factor. \hfill $\Box$
\end{cor}
\tab More on point, \cite[pg.~20]{Jarrow2010} introduced an \emph{event swap} to illustrate the impact of what he deemed to be a ''phantom factor" on portfolio alpha. Specifically, he assumed the possibility of a ''market-systemic" event like a stock market crash, and a financial contract with notional value \$ 1 for which the buyer pays a fixed spread $c$ above the risk free rate $r_t$ for protection $I$ from the seller in the event the ''market-systemic" event occurs. In other words, $c$ is a risk premium. This event is incorporated in the model under the noton of a long position in the ''K-th factor" $X_K(t)$ which is functionally equivalent to a \emph{perpetual event swap} on the market-systemic event. The return on the $K$-th factor was characterized thus:
\begin{equation}
    \text{return on } X_K =
    \begin{cases}
        \frac{-I+X_K(t)}{X_K(t)} & \text{if market-systemic event occurs}\label{eq:ReturnOnKthFactor}\\
        (c+r_t) & \text{otherwise}
    \end{cases}
\end{equation}
Introduction of the market-systemic event alters \ref{eq:JarrowProtterSingleFactorAlpha} by virtue of the $K$-th factor's return, as follows, when $K=1$. Let $\chi_E(t)$ be the characteristic function of the set
\begin{align}
   E &= \{\text{market-systemic event(s)}\}\\
   \alpha(t)dt &= \frac{dS(t)}{S(t)}-(r_t+\beta_1(t)c)dt-\sigma dB(t)-\chi_E(t)\beta_1(t)\biggl(\frac{-I+X_1(t)}{X_1(t)}\biggr)\label{eq:MarketSystemicAlpha}\\
   \intertext{To establish functional equivalence with \ref{eq:CadoganSingleFactorAlpha} we must have the identifying restrictions}
   \sigma &= 1\\
   \frac{dS(t)}{S(t)}&-\chi_E(t)\beta_1(t)\biggl(\frac{-I+X_1(t)}{X_1(t)}\biggr) = 0\\
   S(t) &= S(t-\Delta)\exp\biggl(\int^t_0\chi_E(u)\beta_1(u)\frac{-I+X_1(u)}{X_1(u)}du\biggr),\quad t-\Delta\leq u\leq t\label{eq:AssetPriceJumpOnEvent}\\
   \frac{x}{1-t} &= -(r_t+\beta_1(t)c)\label{eq:AlphaBridgeFactor}
\end{align}
Note that now $B^{br}(t-\Delta)=B^{br}(0)$ where $\Delta$ is the time interval over which assets are priced. This is consistent with futures tradiing where the books are closed at the end of trading day. See \cite[pp.~27-28]{Hull2006}. The foregoing restrictions imply that the asset price $S(t)$ is flat prior to the occurrence of a market systemic event $E$ at time $t$, and its jumps are controlled by $\beta_1(t)$--the trade strategy factor. \cite{BelievaNawalkhaSoto2008} define market systemic events to include ``market crashes, interventions by the Federal Reserve, economic surprises, shocks in the foreign exchange markets, and other rare events". So the characteristic function $\chi_E(t)$ takes the value $1$ for any of those events. \cite[pg.~20]{Jarrow2010} assumed no arbitrage and that $\alpha(t)=0$ in \ref{eq:MarketSystemicAlpha}. Whereupon he concluded that the $\beta_1(t)c$ term is a ``false positive alpha", i.e. an illusionary alpha. Furthermore, he extended his analysis to actively managed portfolios in which the portfolio weights, i.e. factor exposure sensitivity, magnify the so called false positive alpha.

\section{Market-systemic events, interest rate risk, and portfolio alpha}\label{sec:MarketSystemicAlpha}

\tab According to functional equivalence Proposition \autoref{prop:StochasticEquivalenceModels}, in the context of our model, \cite[pg.~20]{Jarrow2010} analysis is problematic for several reasons. First, he \emph{assumed} the existence of a zero alpha event $\mathfrak{O}=\{\omega|\;\alpha_J(t,\omega)=0\}$. However, according to Proposition \ref{prop:CAPM_ZeroSetLocalTime} on page \pageref{prop:CAPM_ZeroSetLocalTime} that event has Lebesgue measure zero. In other words
\begin{align}
   \mathfrak{O}=\{\omega|\;\alpha_J(t,\omega)=0\} &\subseteq \mathcal{Z} = \{\omega|\;\hat{B}(t,\omega)=x\}\\
   \intertext{by virtue of \ref{eq:GammaEstimate}. Thus, $\mathfrak{O}$ is absolutely continuous with respect to $\mathcal{Z}$. Hence}
   Q\{\mathcal{Z}\} = 0 &\Rightarrow Q\{\mathfrak{O}\}=0
\end{align}
According to the Radon-Nykodym Theorem\footnote{See \cite[pg.~78]{GikhmanSkorokhod1969}.}, there exist a unique $\mathcal{Z}$-measureable function $f$ such that
\begin{align}
    \alpha_J(t,\omega) &= \int^t_0f(u)dQ_*(u,\omega)\\
    E\left[\alpha^2_J(t,\omega)\right] &\neq 0
\end{align}
For instance, the probability measure $Q_*$ is the canonical Wiener measure used in \cite{JarrowProtter2010}. So that $\alpha_J(t,\omega)$ is a local $Q_*$-martingale. In which case $f\in \text{L}_2(Q_*)$, the latter being the space of square Lebesgue integrable functions. So the variance of $\alpha_J(t,\omega)$ is non-zero even if $E\left[\alpha_J(t,\omega)\right]=0$. That implies that there times when $\alpha_J$ is above or below $0$. In our case, under Theorem \autoref{thm:GammaRepresentationTheorem} it is a Brownian bridge. The fluctuations around zero are manifest by absolute continuity, which suggests that for \cite[pg.~20]{Jarrow2010} assumed event, we have the probability result $Q(\mathfrak{O})=Q\{\alpha_J(t,\omega) =0 \}=0$ a.s. In other words, Jarrow's alpha is not zero. We summarize the analysis above in the following
\begin{prop}[Non-zero Jarrow alpha]\label{prop:NonZeroJarrowAlpha}~\\
   Let $\mathfrak{O}=\{\omega|\;\alpha_J(t,\omega)=0\}$ be the set of \cite[pg.~20]{Jarrow2010} alpha; $\mathcal{Z} = \{\omega|\;\hat{B}(t,\omega)=x\}$ be the level set at $x$ for $Q$-Brownian motion $\hat{B}$; and $Q_*$ be Wiener measure. There exist $f\in \text{L}_2(Q_*)$ such that $\alpha_J(t,\omega) = \int^t_0f(u)dQ_*(u,\omega)$, \;\;$\mathfrak{O}\subseteq \mathcal{Z}$ and $Q(\mathfrak{O})=0$. \hfill $\Box$
\end{prop}

\tab Second, $\alpha_J(t) =0$ implies $x=0$ in \ref{eq:AlphaBridgeFactor}. Whereupon we have [in an initial equilibrium]
\begin{align}
    &r_t+\beta_1(t)c = 0,\quad r_t >0\label{eq:AlphaZeroCondition}\\
    \intertext{That equation is satisfied under the following scenarios}
    &\text{Scenario 1:} \quad \beta_1(t)>0 \quad \text{and}\;\; c<0\\
    &\text{Scenario 2:} \quad \beta_1(t)<0 \quad \text{and}\;\; c>0
\end{align}
Here, $\;c$ is a risk premium. The stochastic equivalence results summarized in Corollary \autoref{cor:JarrowTradeStrategy} plainly show that $\beta_1(t)$ [and hence $\beta_1(t)c$] is a trading strategy control variable. Therefore, \cite[pg.~20]{Jarrow2010} claim that $\beta_{iK}(t)c$ is a ``false positive alpha" eschews portfolio manager's ``reward" for strategic risk taking. This suggests that the Jarrow-Protter model may be misspecified. For in the context of Scenarios 1 and 2, the following strategies are immediate.

\tab In Scenario 1, our portfolio manager takes a long position, i.e. holds positive amounts, in the ``X-factor" and \emph{receives} a premium $c$--since $c$ is negative--above the risk free rate for doing so. This is similar to a call option strategy based on beliefs about X. So if and when event $E$ occurs, from \ref{eq:AssetPriceJumpOnEvent}, we have $S(t) > S(t-\Delta)$. In Scenario 2, our portfolio manager short sells the ``X-factor" and \emph{pays} a [positive] premium $c$  for the right to do so. This is similar to borrowing $c$ on margin to implement a put option strategy based on beliefs about X. If and when event $E$ occurs $S(t) < S(t-\Delta)$. These option strategies are formulated in the next subsection.

\subsection{The exotic option hedge against market-systemic events}\label{subsec:ExoticOptionHedgeAgainstSystemicRisk}
\tab We consider the market systemic event $E$--even though we do not know when it will happen. Let
\begin{align}
   S_{\text{E}} &\triangleq \{\text{price of asset when market systemic event occurs}\}\\
   K &\triangleq \{\text{strike price of the option such that $K > S_{\text{E}}$}\}\\
   E &= \{\omega|\;S_{\text{E}}(t,\omega|\;\Delta)=\exp\biggl(\int^t_{t-\Delta}\beta_1(u)\frac{-I+X_1(u,\omega)}{X_1(u,\omega)}du\biggr),\;\omega\in\Omega\}\\
   \tau(\omega) &= \inf\{t>0|\;S(t,\omega)\in E\}
\end{align}
The rare event probability associated with $E$ is small and perhaps unmeasureable.
Under Scenario 1, the call option premium is given by
\begin{align}
   &(S(\tau(\omega))-K)^+I_{S_E(t)<K<S(t)}\\
   \intertext{Similarly, under Scenario 2, the put option premium is}
   &(K-S(\tau(\omega))^+I_{S_E(t)<S(t)<K}
\end{align}
where $I$ is an indicator function for the indicated inequalities. Inclusion of the stopping time $\tau(\omega)$ in the option premium implies that they are American style barrier options\footnote{See \cite[pg.~175]{MusielaRutkowski2005}} which can be exercised at or any time after the first occurrence of the systemic event within the specified length of the contract. If the event does not occur, then the contract expires worthless.

\tab For exposition and analytic tractability, assume that in \ref{eq:Self_Fin_Portfolio}, \ref{eq:Factor1Dynamics}, and
\ref{eq:FuncEquivAssetPriceDynamics} the stock price dynamic is a geometric Brownian motion given by
\begin{align}
    \frac{dS(t)}{S(t)} = \mu_S dt &+ \sigma_S dB(t)\\
    \intertext{So that starting at say, time $t-\Delta$, the path properties evolves according to}
    S(t) = S(t-\Delta)\exp\Bigl(\{\mu_S &- \tfrac{\sigma^2_S}{2}\}t+\sigma_SB(t)\Bigr)\\
    \intertext{The probability distribution corresponding to the stopping time is given by}
    P\{S(t) < S_{\text{E}}\} = P\Bigl\{S(t-\Delta)\exp\Bigl(\{\mu_S &-\tfrac{\sigma^2_S}{2}\}t+\sigma_SB(t)\Bigr)<S_{\text{E}}\Bigr\}\\
    \Rightarrow P\{B(t) \leq B_{\text{E}}\} &= \Phi_B(B_{\text{E}}),\;\;\text{where}\\
    B_{\text{E}} = \frac{\ln\bigl(\tfrac{S_{\text{E}}}{S(t-\Delta)}\bigr) - \bigl(\mu_S-\tfrac{\sigma^2_S}{2}\bigr)t}{\sigma_S}&
\end{align}
and $\Phi_B$ is the cumulative normal distribution function for $B(t)\sim N(0,t)$. For the stopping time to exist we need $\Phi_B(B_{\text{E}}) > 0$ which implies that $\sigma_S >0$ and $t < \infty$. We summarize this result in
\begin{prop}[Market Crash Certainty]\label{prop:MarketCrashCertainty}~\\
   As long as asset price volatility exists, the market for a given asset will crash at some point in time. Moreover, the distribution function for the stopping time $\tau(\omega)=\inf\{t>0|\;S(t,\omega)\leq S_{\text{E}}\}$ is given by $\Phi_B(B_{\text{E}})$ where
    \begin{align*}
       \frac{dS(t)}{S(t)} &= \mu_S dt + \sigma_S dB(t)\\
       B_{\text{E}} &= \frac{\ln\bigl(\tfrac{S_{\text{E}}}{S(t-\Delta)}\bigr) - \bigl(\mu_S-\tfrac{\sigma^2_S}{2}\bigr)t}{\sigma_S}
    \end{align*}
and $\Phi_B$ is the distribution function for the Brownian motion $B(t,\omega)$. \hfill $\Box$
\end{prop}
\begin{rem}
   The interested reader is referred to \cite[Lemma~2,~pg.~434]{GikhmanSkorokhod1969} for a more general proof of the existence of finite stopping time for a barrier problem. \hfill $\Box$
\end{rem}
More on point, in this case our portfolio manager would want to buy a \emph{down-and-in put} contract as insurance against adverse market systemic event $E$ if the value of her portfolio falls to a certain level. By substituting $S(t-\Delta)=S(0)$ we can use the closed form formula (with dividend payout rate) in \cite[pp.~534-535]{Hull2006} to obtain the price of this option contract.
\subsection{The swaption hedge against implied interest rate risk caused by market-systemic events}\label{subsec:SwaptionHedgeAgainstMarketSystemic}

\tab Scenarios 1 and 2 imply an asset price ``jump" in \ref{eq:AssetPriceJumpOnEvent} induced by the event $E$, and controlled by the trade strategy variable $\beta_1(t)$. This implies the existence of an option strategy to hedge against the event $E$. In particular, the variables in \ref{eq:AlphaZeroCondition} suggest that our portfolio manager might want to hedge against \emph{de facto} interest rate risk by virtue of movements in $\beta_1(t)$--since we assumed that $c$ is constant and that the risk free rate $r_t$ is relatively deterministic. In other words, there exists a \emph{swaption} strategy to hedge against event $E$. The value of the principal amount of the swap is assumed to be \$1. This suggests that our portfolio manager executes a bond trading strategy to hedge against $E$. The portfolio manager wants to swap the floating rate $\beta_1(t)c$ for the relatively more stable $r_t$. Recall the following conditions for $\beta_1(t)c < 0$:
\begin{enumerate}
  \item [Call:]$\alpha(t) > 0\Rightarrow -\beta_1(t)c-r_t > 0$
  \item [Put:] $\alpha(t) < 0\Rightarrow -\beta_1(t)c-r_t < 0$
\end{enumerate}
Let $P(t,T)$ be the price of a zero coupon bond at time $t$, maturity date $T$ and nominal value $1$; and $N(\cdot)$ be the cumulative normal distribution. Then according to \cite[pp.~626-627]{Hull2006} the price of the swaption starting at time $t$ can be given as follows. Let $\mathcal{T}[t,\infty]$ be the set of all stopping times after time $t$, and
\begin{align}
   g(\beta_1(t)c,r_t) &= (\beta_1(t)c+r_t)^+
\end{align}
be the payoff function for a swap option between $r_t$ and $\beta_1(t)c$. According to \cite[pp.~175,~177]{MusielaRutkowski2005}, the value of the contingent claim at time $t$ is given by the adapted process\footnote{See \cite[pg.~6]{Myneni1992} for existence proof.}
\begin{align}
   V_t &= \text{ess}\;\sup_{\tau\in\mathcal{T}[t,\infty]} E^{P^*}[e^{-r(\tau - t)}g(\beta_1(\tau)c,r_\tau) ]\\
   d_1 &= \frac{\ln\Bigl(\frac{E^P[\beta_1(\tau(\omega))c]} {r_T}\Bigr)+\frac{\sigma^2T}{2}}{\sigma\sqrt{T}}\label{eq:d_1}\\
   d_2 &= d-1 - \sigma\sqrt{T}\label{eq:d_2}\\
   P_{\text{swap}} &= P(0,\tau(\omega))[-r_TN(-d_2)-E^P[\beta_1(\tau(\omega)c]N(-d_1)]\\
   C_{\text{swap}} &= P(0,\tau(\omega))[-E^P[\beta_1(\tau(\omega)c]N(d_1)-r_TN(d_2)]
\end{align}
Here, $P(0,\tau(\omega))$ is a discount factor by virtue of its zero-coupon feature. For the purpose of this paper, all we need to do is establish that $P_{\text{swap}}$ and $C_{\text{swap}}$ exist\footnote{Normally, we would have to evaluate $P(0,\tau(\omega))$ in order to determine the price of the swap. For instance, we could take a risk neutral approach and substitute $\sup_{\tau\in\mathcal{T}[t,\infty]}E^P[P(0,\tau(\omega))]$ instead.  \cite{BelievaNawalkhaSoto2008} show how to price options on zero-coupon bonds using a Vasicek extended jump (Vasicek-EJ) model they attribute to \cite{ChackoDas2002}.}. We begin with the existence of the stopping time.
\begin{lem}[Existence of stopping time for stock market crash]\label{lem:ExistenceStopTimeMarketCrash}~\\
   The stopping time $\tau(\omega) = \inf\{t>0|\;S(t,\omega)\leq S_{\text{E}}\}$ exist with probability 1 for stock price $S(t,\omega)$ and stock price at market market crash $S_{\text{E}}$.
\end{lem}
\begin{proof}
   Apply Proposition \autoref{prop:MarketCrashCertainty} and \cite[Lemma~2]{GikhmanSkorokhod1969} which proves that $\tau(\omega)$ is finite with probability 1.
\end{proof}
Thus we have the following
\begin{prop}[Existence of  positive Jarrow alpha for swaption strategy]~\\
    Let $(\Omega,\mathcal{F}_t,\mathbb{F},P)$ be a filtered probability space, and $E = \{\omega|\;S(t,\omega)\leq S_{\text{E}},\;\omega\in\Omega\}$ be a market systemic event, where $S_E$ is the price of the underlying asset when the event occurs. Let $\tau:\Omega\rightarrow[0,\infty]$ be a mapping that takes non-negative values. Furthermore, let $P(t,T)$ be the price of a zero-coupon bond at time $t$ with maturity date$T$; $r_t$ be the risk free rate; and $\beta_1(t),\;\;0\leq t\leq T$ be a control variable in portfolio manager trade strategy. Let $N(~)$ be the cumulative normal distribution. Then positive alpha exists for the swaption hedge against interest rate risk caused by market systematic event $E$. In particular, the put and call prices for the swaption are give by
    \begin{align*}
       P_{\text{swap}} &= P(0,\tau(\omega))[-r_TN(-d_2)-E^P[\beta_1(\tau(\omega)c]N(-d_1)]\\
       C_{\text{swap}} &= P(0,\tau(\omega))[-E^P[\beta_1(\tau(\omega)c]N(d_1)-r_TN(d_2)]
    \end{align*}
    where $\tau(\omega)=\inf\{t>0|\;S(t,\omega)\leq S_E(t,\omega)\}$ and $E^P$ is the expectations operator with respect to the underlying probability measure $P$, and $d_1$ and $d_2$ are as indicated in \ref{eq:d_1} and \ref{eq:d_2}. \hfill $\Box$
\end{prop}
\begin{proof}
   Under Lemma \ref{lem:ExistenceStopTimeMarketCrash} $P\{\tau(\omega)<\infty\}=1$ a.s. Thus, according to \cite{ChackoDas2002}, and \cite{BelievaNawalkhaSoto2008} extension of \cite{Vasicek1977} model $P(0,\tau(\omega))$ is non-zero. Furthermore, under Corollary \ref{cor:JarrowTradeStrategy} $\beta_1(\tau(\omega)$ is a drift term factor for portfolio manager dynamic [price reversal] trade strategy representation for a single factor asset pricing model like the CAPM. Since each element of the equations for the swaption prices $P_{\text{swap}}$ and $C_{\text{swap}}$ exists, it follows that the swaption prices exist. Moreover, these trade strategies are the result of analytics derived from establishing functional equivalence between \cite{Jarrow2010} and \cite{Cadogan2011b} continuous time models in a single factor setting.
\end{proof}
\section{Conclusion}\label{sec:Conclusion}
This key result in this paper is that \cite{Cadogan2011b} adaptive continuous time representation theorem for alpha, derived from an asymptotic theory of portfolio alpha in a multifactor linear asset pricing model, provides insights about trade strategy whereas \cite{Jarrow2010} continuous time K-factor model does not. In the latter model, alpha is superimposed exogenously, and identifying restrictions are imposed to ascertain the efficacy of alpha. However, results from that approach may not reflect portfolio manager behavior. By contrast, an adaptive model of portfolio manager behavior reflects such behaviour and provides more granular contextual alphas.
\newpage
\appendix
\section*{Appendix}
\addcontentsline{toc}{section}{Appendix}
\section{Matlab code for generating Brownian Bridge}\label{apx:MatlabodeForBB}
\singlespace
\begin{verbatim}
% ******************** Brownian Bridge Function ********************
%% T is matruity time or length of interval
%% N is number of timesteps or inrements
%% dim is dimension, i.e. number of sample funtions, of BB
%%
%
% function bb=BrownianBridge(T,N,dim)
% bb=sqrt(T/N)*cumsum(randn(N,dim));
% bb=bb-bsxfun(@times,(1:N)'/N,bb(end,:))
% bb=[zeros(1,dim);bb];
% end

% ******************** Program Calling Funtion *******************
%
% clear, clf
% T=1
% N=10000
% dim=5
% TestBB=BrownianBridge(T,N,dim) % call Brownian Bridge function with specified parameters
% t=0:(1/N):1
% plot(t,TestBB)
% ylabel('Portfolio Alpha')
% xlabel('Time Interval')
\end{verbatim}
\newpage
\section*{}
\vspace{-10ex}
\bibliographystyle{chicago}        
\addcontentsline{toc}{section}{References} 
\bibliography{MarkeTiming2,EmpiricalAlpha}         

\end{document}